\documentclass[12pt]{article}
\usepackage{amsthm}
\usepackage{amsfonts}
\usepackage{cite}
\newcommand{\cI}{\mathcal{I}}
\newcommand{\cL}{\mathcal{L}}
\newcommand{\cM}{\mathcal{M}}
\newcommand{\cT}{\mathcal{T}}
\newcommand{\FF}{\mathbb{F}}
\newcommand{\NN}{\mathbb{N}}
\newcommand{\QQ}{\mathbb{Q}}
\newcommand{\GF}{\mathbb{GF}}
\def\bw {{\rm bw}}
\def\ind {{\rm ind}}
\def\first#1{{#1}}
\newtheorem{theorem}{Theorem}
\newtheorem{lemma}[theorem]{Lemma}
\newtheorem{corollary}[theorem]{Corollary}
\begin{document}
\title{Decomposition width---a new width parameter for matroids}
\author{Daniel Kr{\'a}l'\thanks{Institute for Theoretical Computer Science (ITI), Faculty of Mathematics and Physics, Charles University, Malostransk{\'e} n{\'a}m{\v e}st{\'\i} 25, 118 00 Prague, Czech Republic. E-mail: {\tt
kral@kam.mff.cuni.cz}.
                                Institute for Theoretical Computer Science (ITI) is supported as project 1M0545 by Czech Ministry of Education.
				This research was also supported by the grant GACR 201/09/0197.}}

\date{}
\maketitle
\begin{abstract}
We introduce a new width parameter for matroids called decomposition width and
prove that every matroid property expressible in the monadic second order logic
can be computed in linear time for matroids with bounded decomposition width
if their decomposition is given.
Since decompositions of small width for our new notion can be computed
in polynomial time for matroids of bounded branch-width represented over finite fields,
our results include recent algorithmic results of Hlin{\v e}n{\'y} [J.~Combin. Theory Ser.~B 96 (2006), 325--351]
in this area and extend his results to matroids not necessarily representable over finite fields.
\end{abstract}

\section{Introduction}
\label{sect-intro}

Algorithmic aspects of graph tree-width form an important part of algorithmic
graph theory. Efficient algorithms for computing tree-decompositions of graphs
have been designed~\cite{bib-arnborg87+,bib-bodlaender93}
and a lot of NP-complete problems become tractable for classes of graphs
with bounded tree-width~\cite{bib-arnborg91+,bib-bodlaender88}.
Most of results of the latter kind are implied
by a general result of Courcelle~\cite{bib-courcelle,bib-courcelle97} that every graph property
expressible in the monadic second-order logic can be decided in linear time
for graphs with bounded tree-width.

As matroids are combinatorial structures that generalize graphs,
it is natural to ask which of these results translate to matroids.
Similarly, as in the case of graphs, some hard problems (that cannot
be solved in polynomial time for general matroids) can be efficiently solved
for (represented) matroids of small width. Though the notion of
tree-width generalizes to matroids~\cite{bib-hlineny-tw,bib-hlineny-tw+},
a more natural width parameter for matroids is the notion of \first{branch-width}.
Let us postpone a formal definition of this width parameter to
Section~\ref{sect-notation} and just mention at this point that
the branch-width of matroids is linearly related with their tree-width,
in particular, the branch-width of a graphic matroid is bounded
by twice the tree-width of the corresponding graph.

There are two main algorithmic aspects which one needs to address
with respect to a width parameter of a combinatorial structure:
\begin{itemize}
\item the efficiency of computing decompositions of small width, and
\item tractability of hard problems for input structures of small width.
\end{itemize}
The first issue has been successfully settled with respect to matroid
branch-width: for every $k$, there exists an algorithm that either computes 
a branch-decomposition of an input matroid with width at most $k$ or
outputs that there is no such branch-decomposition. The first such
algorithm has been found by Oum and Seymour~\cite{bib-oum+} (an approximation
algorithm has been known earlier~\cite{bib-oum-approx}, also
see~\cite{bib-hlineny-fpt-approx} for the case of matroids represented
over a finite field) and
a fixed parameter algorithm was later designed
by Hlin{\v e}n{\'y} and Oum~\cite{bib-hlineny-fpt}.

The tractability results, which include deciding monadic second-order
logic properties~\cite{bib-hlineny-mfcs,bib-hlineny-msol},
computing and evaluating the Tutte polynomial~\cite{bib-hlineny-tutte} and
computing and counting representations over finite fields~\cite{bib-stacs},
usually require restricting to matroids
represented over finite fields (see Section~\ref{sect-notation}
for the definition). This is consistent with the facts that no subexponential
algorithm can decide whether a given matroid is binary~\cite{bib-seymour81},
i.e., representable over $\GF(2)$, even for matroids
with bounded branch-width and that it is NP-hard to decide representability
over $\GF(q)$ for every prime power $q\ge 4$ even for matroids
with bounded branch-width represented over $\QQ$~\cite{bib-hlineny-mfcs06},
as well as
with structural results on matroids~\cite{bib-geelen1+,bib-geelen2+,bib-geelen02+}
that also suggest that matroids representable over finite fields
are close to graphic matroids (and thus graphs) but general matroids can be quite different.

The aim of this paper is to introduce another width parameter for matroids
which will allow to extend the tractability results to matroids not necessarily
representable over finite fields. The cost that needs to be paid for this
is that this new width parameter cannot be bounded by the branch-width in general
though it is bounded by the branch-width for matroids
representable over a fixed finite field. Hence, we can obtain
all tractability results mentioned in the previous paragraph.
The new notion captures the ``structural finiteness''
on cuts represented in the branch decomposition essential
for the tractability results and
is closely related to rooted configurations as introduced
in~\cite{bib-geelen02+} and indistinguishable sets from~\cite{bib-stacs}.

Let us state our results precisely. In Section~\ref{sect-notation},
we introduce a \first{$K$-decomposition} of a matroid (where $K$ is an integer) and
define the \first{decomposition width} of a matroid $\cM$ to be the smallest integer $K$ such that
$\cM$ has a $K$-decomposition.
In Section~\ref{sect-construct},
we show that for every $k$ and $q$, the decomposition width of a matroid with branch-width at most $k$
that is representable over $\GF(q)$ is at most $F(k,q)$ and
that $F(k,q)$-decomposition of any such matroid can be computed in polynomial time.
In Sections~\ref{sect-tutte}, \ref{sect-msol} and \ref{sect-repr}, we show that
for every $K$, there exist polynomial-time algorithms (with the degree of the polynomial
independent of $K$) for computing and evaluating the Tutte polynomial,
deciding monadic second-order logic properties, deciding representability and
constructing and counting representations over finite fields when the input matroid
is given by its $K$-decomposition. In particular, our results imply all the tractability
results known for matroids represented over finite field (we here claim only the polynomiality,
not matching the running times which we did not try to optimize throughout the paper).

A $K$-decomposition of a matroid actually captures the whole structure of a matroid, i.e., the matroid
is fully described by its $K$-decomposition, and thus it can be understood as an alternative way
of providing the input matroid. In fact, for a fixed $K$, the size of a $K$-decomposition is linear
in the number of matroid elements and thus this representation of matroids is very suitable
for this purpose. Let us state (without a proof) that $K$-decompositions of matroids {\em supports}
contraction and deletion of matroid elements without increasing the decomposition width as well as
some other matroid operations with increasing the decomposition width by a constant,
e.g., relaxing a circuit-hyperplane increases the decomposition width by at most one.
By {\em supporting} we mean that a $K$-decomposition of the new matroid can be efficiently
computed from the $K$-decomposition of the original one. Hence, the definition of the decomposition
width does not only yield a framework extending tractability results for matroids represented
over finite fields with bounded branch-width but it also yields a compact data structure
for representing input matroids.

\section{Notation}
\label{sect-notation}

In this section, we formally introduce the notions
which are used in this paper. We start with basic notions and
we then introduce matroid representations, branch-decompositions and
our new width parameter. We also refer the reader
to the monographs~\cite{bib-oxley,bib-truemper}
for further exposition on matroids.

A \first{matroid $\cM$} is a pair $(E,\cI)$
where $\cI\subseteq 2^E$. The elements of $E$ are
called \first{elements} of $\cM$, $E$ is the \first{ground set} of $\cM$ and
the sets contained in $\cI$ are called \first{independent} sets.
The set $\cI$ is required to contain the empty set, to be hereditary, i.e.,
for every $F\in\cI$, $\cI$ must contain all subsets of $F$, and
to satisfy the exchange axiom: if $F$ and $F'$ are two sets of $\cI$ such
that $|F|<|F'|$, then there exists $x\in F'$ such that $F\cup\{x\}\in\cI$.
The \first{rank} of a set $F$,
denoted by $r(F)$, is the size
of the largest independent subset of $F$ (it can be inferred from the exchange
axiom that all inclusion-wise maximal independent subsets of $F$ have the same size).
In the rest, we often understand matroids as sets of elements
equipped with a property of ``being independent''. We use $r(\cM)$
for the rank of the ground set of $E$.

If $F$ is a set of elements of $\cM$, then $\cM\setminus F$
is the matroid obtained from $\cM$ by \first{deleting} the elements of $F$,
i.e., the elements of $\cM\setminus F$ are those not contained in $F$ and
a subset $F'$ of such elements is independent in the matroid $\cM\setminus F$
if and only if $F'$ is independent in $\cM$. The matroid $\cM/F$
which is obtained by \first{contraction} of $F$ is the following matroid:
the elements of $\cM/F$ are those not
contained in $F$ and a subset $F'$ of such elements is independent
in $\cM/F$ if and only if $r(F\cup F')=|F|+|F'|$.
A \first{loop} of $\cM$ is an element $e$ of $\cM$ such that $r(\{e\})=0$ and
a \first{bridge} is an element such that $r(\cM\setminus\{e\})=r(\cM)-1$.
A \first{separation $(A,B)$} is a partition of the elements of $\cM$
into two disjoint sets and a separation is called a \first{$k$-separation}
if $r(\cM)-r(A)-r(B)\le k-1$.

\subsection{Matroid representations}

Matroids do not generalize only the notion of graphs but they also generalize
the notion of linear independence of vectors. If $\FF$ is a (finite or infinite) field,
a mapping $\varphi:E\to\FF^d$ from the ground set $E$ of $\cM$
to a $d$-dimensional vector space over $\FF$ is a \first{representation}
of $\cM$ if a set $\{e_1,\ldots,e_k\}$ of elements of $\cM$
is independent in $\cM$ if and only if $\varphi(e_1),\ldots,\varphi(e_k)$
are linearly independent vectors in $\FF^d$. For a subset $E$
of the elements of $\cM$, $\varphi(E)$ denotes the linear subspace
of $\FF^d$ generated by the images of the elements of $E$.
In particular, $\dim\varphi(E)=r(E)$. Two representations $\varphi_1$ and
$\varphi_2$ of $\cM$ are isomorphic if there exists an isomorphism $\psi$
of vector spaces $\varphi_1(\cM)$ and $\varphi_2(\cM)$ such that
$\psi(\varphi_1(e))$ is a non-zero multiple of $\varphi_2(e)$ for every
element $e$ of $\cM$.

We next introduce additional notation for vector spaces over a field $\FF$.
If $U_1$ and $U_2$ are two linear subspaces of a vector space
over $\FF$, $U_1\cap U_2$ is the linear space
formed by all the vectors lying in both $U_1$ and $U_2$, and
$\cL(U_1\cup U_2)$ is the linear space formed by all the linear
combinations of the vectors of $U_1$ and $U_2$, i.e., the linear hull
of $U_1\cup U_2$.
Formally, $v\in\cL(U_1\cup U_2)$ if and only if there exist
$v_1\in U_1$, $v_2\in U_2$ and $\alpha_1,\alpha_2\in\FF$ such that
$v=\alpha_1 v_1+\alpha_2 v_2$.

\subsection{Branch-decompositions}
\label{subsect-branch}

A \first{branch-decomposition} of a matroid $\cM$ with ground set $E$ is a tree $T$ such that
\begin{itemize}
\item all inner nodes of $T$ have degree three, and
\item the leaves of $T$ one-to-one correspond to the elements of $\cM$.
\end{itemize}
An edge $e$ of $T$ splits $T$ into two subtrees and the elements
corresponding to the leaves of the two subtrees form a partition
$(E_1,E_2)$ of the ground set $E$. The \first{width of an edge}
$e$ is equal to $r(E_1)+r(E_2)-r(E)$, i.e., to the smallest $k$
such that $(E_1,E_2)$ is a $(k+1)$-separation of $\cM$. The
\first{width of the branch-decomposition} $T$ is the maximum width
of an edge $e$ of $T$. Finally, the
\first{branch-width}\index{branch-width} of a matroid is the
minimum width of a branch-decomposition of $\cM$ and is denoted by
$\bw(\cM)$.

\subsection{Decomposition width}
\label{subsect-decomp}

We now formally define our new width parameter.
A \first{$K$-decompo\-sition}, $K\ge 1$,
of a matroid $\cM$ is a rooted tree $\cT$ such that
\begin{itemize}
\item the leaves of $\cT$ one-to-one correspond to the elements of $\cM$ and
      carry the information whether the associated element of $\cM$ is a loop, and
\item each inner node $v$ of $\cT$ has exactly two children and is associated with two functions $\varphi_v$ and $\varphi_v^r$ such that
      $\varphi_v:\{0,\ldots,K\}^2\to\{0,\ldots,K\}$ and $\varphi_v^r:\{0,\ldots,K\}^2\to\NN$.
\end{itemize}
The tree $\cT$ represents the rank function of $\cM$ in the way we
now describe. For a subset $F$ of the ground set of $\cM$, we
label and color the vertices of $\cT$ in the following way: a leaf
of $\cT$ corresponding to an element in $F$ is colored with $1$
and other leaves are colored with $0$. The leaves colored with $1$
corresponding to non-loop elements of $\cM$ are labeled with $1$
and all other leaves are labeled with $0$. If $v$ is an inner node
of $\cT$ and its two children are labeled with $\lambda_1$ and
$\lambda_2$ and colored with $\gamma_1$ and $\gamma_2$, the node
$v$ is colored with $\varphi_v(\gamma_1,\gamma_2)$ and labeled
with the number
$\lambda_1+\lambda_2-\varphi_v^r(\gamma_1,\gamma_2)$. The rank of
$F$ is the label of the root of $\cT$. The \first{decomposition
width} of a matroid $\cM$ is the
smallest $K$ such that $\cM$ has a $K$-decomposition.

Let us give an intuitive explanation of the above procedure. The
colors represent types of different subsets of elements of $\cM$,
i.e., if $E_v$ is the set of elements assigned to leaves of a subtree
rooted at an inner vertex $v$, then those subsets of $E_v$ that get
the same color at $v$ are of the same type. The labels represent
the rank of subsets. Subsets of the same type can have different
labels (and thus ranks) but they behave in the same way in the
following sense: if $E_1$ and $E_2$ are elements assigned to
leaves of two subtrees rooted at children of $v$, then the rank of
the union of two subsets of $E_1$ with the same color and two
subsets of $E_2$ with the same color is equal to the sum of their
ranks decreased by the same amount.

In what follows, we will always assume that if $F$ is the empty
set, then all the nodes of $\cT$ are colored and labelled with $0$ and
consider this assumption to be part of the definition of a $K$-decomposition.

\section{Constructing decompositions}
\label{sect-construct}

In this section, we relate the decomposition width of matroids representable
over finite fields to their branch-width.

\begin{theorem}
\label{thm-construct}
Let $\cM$ be a matroid representable over a finite field $\FF$ of order $q$.
If the branch-width of $\cM$ is at most $k\ge 1$,
then the decomposition width of $\cM$ is at most $K=\frac{q^{k+1}-q(k+1)+k}{(q-1)^2}$
Moreover, if a branch-decomposition of $\cM$ with width $k$ and its representation over $\FF$ are given,
then a $K$-decomposition of $\cM$ can be constructed in time $O(n^{1+\alpha})$
where $n$ is the number of elements of $\cM$ and $\alpha$ is the exponent from the matrix multiplication algorithm.
\end{theorem}

\begin{proof}
Let $\cT_b$ be the branch-decomposition of $\cM$ of width at most
$k$. Subdivide an arbitrary edge of $\cT_b$ and root the resulting
tree $\cT$ at the new vertex. We now have to construct the
functions $\varphi_v$ and $\varphi_v^r$. Fix a representation of
$\cM$ over $\FF$ and let $w_e$ be a vector representing an element
$e$ of $\cM$. For a node $v$ of $\cT$, let $W_v$ be the linear
hull of all the vectors $w_e$ representing the elements $e$
corresponding to the leaves of the subtree rooted at $v$ and
$W'_v$ be the linear hull of the remaining vectors $w_e$, i.e.,
vectors $w_e$ for the elements $e$ not corresponding to the leaves
of the subtree. Finally, let $D_v$ be the intersection of $W_v$
and $W'_v$. Since $\cT_b$ is a branch-decomposition of width at
most $k$, the dimension of the subspace $D_v=W_v\cap W'_v$ is at
most $k$ for every inner node $v$.

Each of at most $K+1=1+\sum_{i=1}^k\frac{q^i-1}{q-1}$ subspaces of
$D_v$ is associated with one of the colors $\{0,\ldots,K\}$ in
such a way that the trivial subspace of dimension $0$ is
associated with the color $0$. If $v$ is a leaf associated with
non-trivial subspace $D_v$ (the corresponding element is neither
loop nor co-loop), then the subspace $D_v$ itself is associated
with the color $1$. Let $v$ be an inner node of $\cT$ with
children $v_1$ and $v_2$. For $\gamma_i\in\{0,\ldots,K\}$, $i=1,2$, let
$\Gamma_i$ be the subspace of $W_i$ colored with $\gamma_i$. The function
$\varphi_v(\gamma_1,\gamma_2)$ is equal to the color $z$ of the subspace
$D_v\cap\cL(\Gamma_1\cup \Gamma_2)$ where $\cL(\Gamma_1\cup \Gamma_2)$ is the linear
hull of $\Gamma_1\cup \Gamma_2$ and $\varphi_v^r(\gamma_1,\gamma_2)$ is equal
$$\dim \Gamma_1+\dim \Gamma_2-\dim\cL(\Gamma_1\cup \Gamma_2)\;\mbox{.}$$
If there is no subspace of $D_{v_i}$ with color $\gamma_i$ for $i=1$ or $i=2$, then
the functions $\varphi_v(\gamma_1,\gamma_2)$ and $\varphi_v^r(\gamma_1,\gamma_2)$ are
equal to $0$. Finally, if $v$ is the root of the tree $\cT$, then
the function $\varphi_v^r$ is defined in the same way and
$\varphi_v$ is defined to be constantly equal to $0$.

We now have to verify that the constructed decomposition of $\cM$
represents the matroid $\cM$. Let $F$ be a subset of elements of
$\cM$ and consider the coloring and the labeling of the nodes of
$\cM$ as in the definition of a $K$-decomposition. We will prove
that the label of each node $v$ of $\cT$ is equal to the dimension
of the linear hull $X_v$ of the vectors $w_e$ with $e \in F_v$ where
$F_v$ is the set of the elements of $F$ corresponding to the
leaves of the subtree rooted at $v$.

The label of each leaf $v$ is equal to the dimension of $X_v$
since the leaves corresponding to
elements $e$ of $F$ that are not loops are labeled with $1$ and
the other leaves are labeled with $0$. Let $v$ be a node of $\cT$
with children $v_1$ and $v_2$. Let us compute the dimension of
$X_v$:
\begin{eqnarray*}
\dim X_v & = & \dim X_{v_1}+\dim X_{v_2}-\dim X_{v_1}\cap X_{v_2} \\
         & = & \dim X_{v_1}+\dim X_{v_2}-\dim X_{v_1}\cap X_{v_2}\cap D_{v_1}\cap D_{v_2} \\
         & = & \dim X_{v_1}+\dim X_{v_2}+\\
     & & \dim X_{v_1}\cap D_{v_1}+\dim X_{v_2}\cap D_{v_2}-\\
     & & \dim \cL\left((X_{v_1}\cap D_{v_1})\cup(X_{v_2}\cap D_{v_2})\right)
\end{eqnarray*}
since $X_{v_1}\cap X_{v_2}\subseteq W_{v_i}\cap W'_{v_i}=D_{v_i}$. As
$$\varphi_v^r(\gamma_1,\gamma_2)=\dim X_{v_1}\cap D_{v_1}+\dim X_{v_2}\cap D_{v_2}-\dim \cL\left((X_{v_1}\cap D_{v_1})\cup(X_{v_2}\cap D_{v_2})\right)$$
where $\gamma_i$ is the color of $X_{v_i}\cap D_{v_i}$, the label of $v$ is equal to the dimension of $X_v$.
This finishes the proof that the decomposition $\cT$ represents the matroid $\cM$.

It remains to explain how to construct the decomposition $\cT$ within the claimed time. The tree $\cT$ has $n-1$ non-leaf nodes $v$ and
we need time $O(n^\alpha)$ to compute the subspaces $W_v$, $W'_v$ and $D_v$. Since the dimension of $D_v$ is bounded by $k$,
associating colors $\{0,\ldots,K\}$ with the subspaces of $D_v$ and defining the functions $\varphi_v$ and $\varphi^r_v$
requires time bounded by $K$ and $q$ for each node $v$.
\end{proof}

Theorem~\ref{thm-construct} and
the cubic-time algorithm of Hlin{\v e}n{\'y} and Oum~\cite{bib-hlineny-fpt}
for computing branch-decompositions of matroids with bounded branch-width yield the following:

\begin{corollary}
\label{cor-construct}
Let $\cM$ be a matroid represented over a finite field $\FF$ of order $q$.
For every $k\ge 1$, there exists an algorithm running in time $O(n^{1+\alpha})$,
where $n$ is the number of elements of $\cM$ and $\alpha$ is the exponent from the matrix multiplication algorithm,
that either outputs that the branch-width of $\cM$ is bigger than $k$ or
construct a $K$-decomposition of $\cM$ for $K\le\frac{q^{k+1}-q(k+1)+k}{(q-1)^2}$.
\end{corollary}

\section{Verifying decompositions}
\label{sect-verify}

A $K$-decomposition fully describes a considered matroid $\cM$
through functions $\varphi_v$ and $\varphi^r_v$ for internal
nodes $v$ of the decomposition. Clearly, not all possible
choices of $\varphi_v$ and $\varphi^r_v$ give rise to a decomposition
representing a matroid and it is natural to ask whether we can
efficiently test that a $K$-decomposition represents a matroid.
We answer this equation in the affirmative way in the next theorem.

\begin{theorem}
\label{thm-verify}
For every $K$-decomposition $\cT$ with $n$ leaves, it can be tested
in time $O(K^8n)$ whether $\cT$ corresponds to a matroid.
\end{theorem}

\begin{proof}
One of the equivalent definitions of a matroid is given
by the submodularity of its rank function, i.e.,
a function $r:2^E\to\NN$ is a rank function of a matroid if and only if 
\begin{equation}
r(A\cup B)+r(A\cap B)\le r(A)+r(B)\label{eq-sub}
\end{equation}
for every subsets $A$ and $B$ of $E$.
We will verify using the $K$-decomposition that the inequality (\ref{eq-sub})
holds for all $A,B\subseteq E$.

We first check that for the empty set, all nodes of $\cT$ are colored with $0$ and
labelled with $0$. If this is not the case, we reject.
Let $E_v$ be now the set of the elements assigned to the leaves of a subtree of $\cT$
rooted at $v$.
In order to achieve our goal, we will compute for every node $v$ of $\cT$ and
every quadruple $[\gamma_A,\gamma_B,\gamma_\cap,\gamma_\cup]\in\{0,\ldots,K\}^4$
the minimum possible difference
$$\lambda(A)+\lambda(B)-\lambda(A\cup B)-\lambda(A\cap B)$$
where $A$ is a subset of $E_v$ colored at $v$ with $\gamma_A$ and
$B$ a subset colored with $\gamma_B$ such that the color of $A\cup B$ is $\gamma_\cup$ and
the color of $A\cap B$ is $\gamma_\cap$ and $\lambda(A)$, $\lambda(B)$, $\lambda(A\cup B)$ and $\lambda(A\cap B)$
are the labels corresponding to the subsets $A$, $B$, $A\cup B$ and $A\cap B$ at $v$
when the coloring and labelling procedure from the definition of a $K$-decomposition
is applied. This minimum is further denoted by $\mu_v(\gamma_A,\gamma_B,\gamma_\cap,\gamma_\cup)$.
If there is no pair of subsets $A$ and $B$ such that the color of $A$ at $v$ is $\gamma_A$,
the color of $B$ is $\gamma_B$, the color of $A\cup B$ is $\gamma_\cup$ and the color of $A\cap B$ is $\gamma_\cap$,
then $\mu_v(\gamma_A,\gamma_B,\gamma_\cap,\gamma_\cup)=\infty$.

Let us see how the values of $\mu$ can be computed for the nodes of $\cT$ in the direction
from the leaves towards the root of $\cT$. If $v$ is a leaf of $\cT$,
then the function $\mu_v$ is equal to zero for the quadruples $[0,0,0,0]$, $[1,0,0,1]$,
$[0,1,0,1]$ and $[1,1,1,1]$ and is equal to $\infty$ for all other quadruples.
If $v$ is an inner node of $\cT$ with children $v'$ and $v''$,
then $\mu_v(\gamma_A,\gamma_B,\gamma_\cap,\gamma_\cup)$ is equal to the minimum of 
$$
\mu_{v'}(\gamma'_A,\gamma'_B,\gamma'_\cap,\gamma'_\cup)+
\mu_{v''}(\gamma''_A,\gamma''_B,\gamma''_\cap,\gamma''_\cup)-$$
$$\varphi_v^r(\gamma'_A,\gamma''_A)
-\varphi_v^r(\gamma'_B,\gamma''_B)
+\varphi_v^r(\gamma'_\cap,\gamma''_\cap)
+\varphi_v^r(\gamma'_\cup,\gamma''_\cup)$$
where the minimum is taken over all pairs of quadruples $[\gamma'_A,\gamma'_B,\gamma'_\cap,\gamma'_\cup]$ and
$[\gamma''_A,\gamma''_B,\gamma''_\cap,\gamma''_\cup]$ such that
$\varphi_v(\gamma'_A,\gamma''_A)=\gamma_A$,
$\varphi_v(\gamma'_B,\gamma''_B)=\gamma_B$,
$\varphi_v(\gamma'_\cap,\gamma''_\cap)=\gamma_\cap$ and
$\varphi_v(\gamma'_\cup,\gamma''_\cup)=\gamma_\cup$.
Here, we assume that $\infty+k=\infty$ and $k<\infty$ for any $k\in\NN$.
It is straightforward to verify that the function $\mu_v$ computed
in this way is equal to the claimed minimum.

The function $\mu_v$ for the root $v$ of $\cT$ determines whether the $K$-de\-compo\-si\-tion
represents a matroid. If the value of $\mu_v$ for any quadruple is negative,
then there exists subsets $A$ and $B$ violating (\ref{eq-sub}) and thus
the rank function given by $\cT$ is not submodular.
On the other hand, if all the values of $\mu_v$ are non-negative or
equal to $\infty$ for all quadruples, then (\ref{eq-sub}) holds
for every two subsets $A$ and $B$. Consequently, $\cT$ represents a matroid.

It remains to estimate the running time of the algorithm.
Clearly, the minimums can be computed in time $O(K^8)$
for every inner node of $\cT$. We conclude that
the running time of the whole algorithm is $O(K^8n)$ as claimed.
\end{proof}

\section{Computing the Tutte polynomial}
\label{sect-tutte}

One of the classical polynomials associated to matroids is the Tutte polynomial. There are several
equivalent definitions of this polynomial but we provide here only the one we use.
For a matroid $\cM$ with ground set $E$,
the \first{Tutte polynomial} $T_{\cM}(x,y)$ is equal to
\begin{equation}
T_{\cM}(x,y)=\sum_{F\subseteq E} (x-1)^{r(E)-r(F)}(y-1)^{|F|-r(F)}
\label{eqd-14}
\end{equation}
The Tutte polynomial is an important algebraic object associated to a matroid.
Some of the values of $T_{\cM}(x,y)$ have a combinatorial interpretation;
as a simple example,
the value $T_{\cM}(1,1)$ is equal to the number of bases of a matroid $\cM$.
We show that the Tutte polynomial can be computed and evaluated
in time $O(K^2n^3r^2)$ for $n$-element matroids of rank $r$
given by the $K$-decomposition. The part of the proof for computing
the Tutte polynomial reflects the main motivation behind the definition
of matroid decompositions given in the previous section.

\begin{theorem}
\label{thm-tutte}
Let $K$ be a fixed integer. The Tutte polynomial of
an $n$-element matroid $\cM$ given by its $K$-decomposition can be
computed in time $O(K^2n^3r^2)$ and evaluated in time $O(K^2n)$
where $r$ is the rank of $\cM$ (under the assumption that summing
and multiplying $O(n)$-bit numbers can be done in constant time).
\end{theorem}

\begin{proof}
First, we give an algorithm for computing the Tutte polynomial,
i.e., an algorithm that computes the coefficients in the polynomial.
Let $\cT$ be a $K$-decomposition of $\cM$ and let $E_v$ be the elements
of $\cM$ corresponding to leaves of the subtree rooted at a vertex $v$.
For every node $v$ of $\cM$ and every triple $[\gamma,n',r']$,
$0\le \gamma\le K$, $0\le n'\le n$ and $0\le r'\le r$,
we compute the number of subsets $F$ of $E_v$
such that the color assigned to $v$ for $F$ is $\gamma$, $|F|=n'$ and
the rank of $F$ is $r'$. These numbers will be denoted by $\mu_v(\gamma,n',r')$.

We compute the numbers $\mu_v(\gamma,n',r')$ from the leaves to the root of $\cT$.
If $v$ is a leaf of $\cT$, then $\mu_v(0,1,0)$ is equal to $2$ if $v$ corresponds to a loop of $\cM$ and
it is equal to $1$, otherwise. In the latter case, i.e., if $v$ corresponds to a non-loop,
$\mu_v(1,1,1)$ is also equal to $1$. All the other values of $\mu_v$ are equal to $0$.

Let $v$ be a node with children $v_1$ and $v_2$. If $F_i$ is a subset of $E_{v_i}$
with color $\gamma_i$ such that $n_i=|F_i|$ and $r_i=r(F_i)$, then $F_1\cup F_2$
is a subset of $E_v$ with color $\varphi_v(\gamma_1,\gamma_2)$, with $n_1+n_2$ elements and
the rank $r_1+r_2-\varphi_v^r(\gamma_1,\gamma_2)$. In other words, it holds that
\begin{equation}
\mu_v(\gamma,n',r')=\sum_{\gamma_1,n_1,r_1,\gamma_2,n_2,r_2}\mu_{v_1}(\gamma_1,n_1,r_1)\mu_{v_2}(\gamma_2,n_2,r_2)\label{eqd-15}
\end{equation}
where the sum is taken over six-tuples $(\gamma_1,n_1,r_1,\gamma_2,n_2,r_2)$ such that
$\gamma=\varphi_v(\gamma_1,\gamma_2)$, $n'=n_1+n_2$ and $r'=r_1+r_2-\varphi_v^r(\gamma_1,\gamma_2)$.
Computing $\mu_v$ from the values of $\mu_{v_1}$ and $\mu_{v_2}$
base on (\ref{eqd-15}) requires time $O(K^2n^2r^2)$ and
thus the total running time of the algorithm is $O(K^2n^3r^2)$.
The Tutte polynomial of $\cM$ can be read from $\mu_r$ where $r$ is the root of $\cT$
since the value $\mu_r(0,\alpha,\beta)$
is the coefficient at $(x-1)^{r(E)-\beta} (y-1)^{\alpha-\beta}$ in (\ref{eqd-14}).

Let us turn our attention to evaluating the Tutte polynomial for given values of $x$ and $y$.
This time, we recursively compute the following quantity for every node $v$ of $\cT$:
\begin{equation}
\mu_v(\gamma)=\sum_{F\subseteq E_v}(x-1)^{r(E)-r(F)}(y-1)^{|F|-r(F)}
\label{eqd-16}
\end{equation}
where the sum is taken over the subsets $F$ with color $\gamma$.
The value of $\mu_v(0)$ is equal to $(x-1)^{r(E)}y$ for a leaf $v$ corresponding to a loop of $\cM$ and
to $(x-1)^{r(E)}$ if $v$ correspond to a non-loop. In the latter case, $\mu_v(1)$ is also equal
to $(x-1)^{r(E)-1}$. All other values of $\mu_v$ are equal to $0$ in both cases.

For a node $v$ of $\cT$ with two children $v_1$ and $v_2$, the equation (\ref{eqd-16}) and
the definition of a $K$-decomposition implies that
$$\mu_v(\gamma)=\sum_{0\le \gamma_1,\gamma_2\le K}\frac{\mu_{v_1}(\gamma_1)\mu_{v_2}(\gamma_2)}{(x-1)^{r(E)+\varphi^r_v(\gamma_1,\gamma_2)}(y-1)^{\varphi_v^r(\gamma_1,\gamma_2)}}$$
where the sum is taken over values $\gamma_1$ and $\gamma_2$ such that $\gamma=\varphi_v(\gamma_1,\gamma_2)$.
Under the assumption that arithmetic operations require constant time, determining the values
of $\mu_v$ needs time $O(K^2)$. Since the number of nodes of $\cT$ is $O(n)$, the total running
time of the algorithm is $O(K^2n)$ as claimed in the statement of the theorem.
\end{proof}

As a corollary, we obtain Hlin{\v e}n{\'y}'s result on computing the Tutte polynomial and its values
for matroids represented over finite fields of bounded branch-width~\cite{bib-hlineny-tutte}.

\begin{corollary}
\label{cor-tutte}
Let $\FF$ be a fixed finite field and $k$ a fixed integer. There is a polynomial-time algorithm
for computing and evaluating the Tutte polynomial for the class of matroids of branch-width at most $k$ representable over $\FF$ that are given by their representation over the field $\FF$.
\end{corollary}

\begin{proof}
Let $\cM$ be a matroid of branch-width at most $k$ represented over $\FF$.
By Corollary~\ref{cor-construct}, we can efficiently find a $K$-decomposition of $\cM$
where $K$ is a constant depending only on $\FF$ and $k$.
As $K$ depends on $\FF$ and $k$ only, the Tutte polynomial of $\cM$ can
be computed and evaluated in time polynomial in the number of elements of $\cM$ by Theorem~\ref{thm-tutte}.
\end{proof}

\section{Deciding MSOL-properties}
\label{sect-msol}

In this section, we show that there is a linear-time
algorithm for deciding monadic second order logic formulas for
matroids of bounded decomposition width. First, let us be precise
on the type of formulas that we are interested in. A
\first{monadic second order logic formula} $\psi$, an MSOL formula, for a matroid
contains basic logic operators (the negation, the disjunction and
the conjunction), quantifications over elements and sets of
elements of a matroid (we refer to these variables as to element
and set variables), the equality predicate, the predicate of
containment of an element in a set and the independence predicate
which determines whether a set of elements of a matroid is
independent. The independence predicate depends on and encodes an
input matroid $\cM$.

In order to present the algorithm, we introduce auxiliary notions
of a $K$-half and $K$-halved matroids which we extend to
interpreted $K$-halves. A \first{$K$-half} is a
matroid $\cM$ with ground-set $E$ and equipped with a function
$\varphi_{\cM}: 2^{E}\to\{0,\ldots,K\}$. A
\first{$K$-halved} matroid is a matroid
$\cM$ with ground-set $E=E_1\cup E_2$ composed of a $K$-half
$\cM_1$ with ground set $E_1$ and a matroid $\cM_2$ with ground
set $E_2$ such that each subset of $F\subseteq E_2$ is assigned a
vector $w^F$ of $K+1$ non-negative integers. Both $\cM_1$ and
$\cM_2$ are matroids. The rank of a subset $F\subseteq E_1\cup
E_2$ is given by the following formula: $$r_{\cM}(F) =
r_{\cM_1}(F\cap E_1) + r_{\cM_2}(F\cap E_2) - w^{F\cap
E_2}_{\varphi_{\cM_1}(F\cap E_1)}\;\mbox{,}$$ where $w^{F\cap
E_2}_{\varphi_{\cM_1}(F\cap E_1)}$ is the coordinate
of the vector $w^{F\cap E_2}$ corresponding to $\varphi_{\cM_1}(F\cap E_1)$.
We will write $\cM=\cM_1\oplus_K \cM_2$ to represent the fact that
the matroid $\cM$ is a $K$-halved matroid obtained from $\cM_1$
and $\cM_2$ in the way we have just described.

The next lemma justifies the just introduced definition of $K$-halved matroids
since it asserts that every matroid $\cM$ represented by a $K$-decomposi\-tion
can be viewed as a composed of two $K$-halves, one of which
is $\cM$ restricted to the elements corresponding to leaves of
a subtree of its $K$-decomposition.

\begin{lemma}
\label{lm-aux}
Let $\cT$ be a $K$-decomposition of a matroid $\cM$, $K\ge 1$, and let $v$ be a node of $\cT$.
Further, let $E_v$ be the set of elements of $\cM$ assigned to the leaves of the subtree of $\cT$ rooted at $v$, and
$\overline{E_v}$ the set of the remaining elements of $\cM$. If $F_1$ and $F_2$ are two subsets of $E_v$ such that
the color of $v$ with respect to the decomposition is the same for $F_1$ and $F_2$, then
$$r(F)+r(F_1)-r(F\cup F_1)=r(F)+r(F_2)-r(F\cup F_2)$$
for every subset $F$ of $\overline{E_v}$.
\index{$K$-decomposition}
\end{lemma}

\begin{proof}
The equality $$r(F)+r(F_1)-r(F\cup F_1)=r(F)+r(F_2)-r(F\cup F_2)$$
follows from the fact that the vertex $v$ gets the same color for
both sets $F_1$ and $F_2$: switching between $F_1$ and $F_2$ can
change the labels of the vertices in the subtree of $v$ and on the
path $P$ from $v$ to the root only. Since the color of $v$ is the
same for $F_1$ and $F_2$, switching between $F_1$ and $F_2$ does
not change the colors of the vertices on the path $P$. Hence, the
label of the root is equal to the sum of the labels of the
children of the vertices on $P$ decreased by the values of the
functions $\varphi_w$ for vertices on $P$ which are not affected
by switching between $F_1$ and $F_2$ since they only depend on the
colors of children of $w$ which are not changed by the switch.
Hence, the equality from the statement of the lemma holds.
\end{proof}

For the purpose of induction in the next lemma, let us allow free
variables in formulas. For an MSOL formula $\psi$ with free
variables $\xi_1,\ldots,\xi_k$, a \first{$K$-signature} $\sigma$
is a mapping $\sigma:\{1,\ldots,k\}\to\{0,\ldots,K\}\cup\{\star\}$
such that $\sigma^{-1}(\star)$ contains only indices of element variables.
A \first{$\sigma$-interpretation} of a $K$-half $\cM_1$
is an assignment of elements $e$ to element variables $\xi_j$ with
$\sigma(j)\not=\star$ such that $\varphi_{\cM_1}(\{e\})=\sigma(j)$
and an assignment of subsets $F$ of the ground set of $\cM_1$ to
set variables $\xi_j$ with $\varphi_{\cM_1}(F)=\sigma(j)$. A
\first{$\sigma$-interpretation} of $\cM_2$ is an assignment of
elements of $\cM_2$ to element variables $\xi_j$ with
$\sigma(j)=\star$ and an assignment of subsets $F$ of the ground
set of $\cM_2$ to set variables $\xi_j$.
A matroid $\cM_i$, $i\in\{1,2\}$,
with a $\sigma$-interpretation is said to be {\em $\sigma$-interpreted}.

If both $\cM_1$ and $\cM_2$ are $\sigma$-interpreted, then the $\sigma$-interpretations of $\cM_1$ and $\cM_2$
naturally give an assignment of element variables $\xi_j$ (given by the assignment for $\cM_1$
if $\sigma(j)\not=\star$ and by the assignment for $\cM_2$, otherwise) and an assignment of set variables $\xi_j$
by uniting the assignments for $\xi_j$ in the $\sigma$-interpretations of $\cM_1$ and $\cM_2$.
For a formula $\psi$ with a $K$-signature $\sigma$,
$\psi$ is satisfied for a matroid $\cM=\cM_1\oplus_K \cM_2$
with a $\sigma$-interpreted $K$-half $\cM_1$ and $\sigma$-interpreted $\cM_2$
if the assignment given by the $\sigma$-interpretations and $\cM$ satisfies $\psi$.

The crucial notion in our argument is that of $(\psi,\sigma)$-equivalence of $\sigma$-inter\-pre\-ted $K$-halves
which we now define. For a formula $\psi$ with a $K$-signature $\sigma$,
two $\sigma$-interpreted $K$-halves $\cM_1$ and $\cM'_1$ are {\em $(\psi,\sigma)$-equivalent}
if for every $\sigma$-interpreted $\cM_2$, the formula $\psi$ is satisfied for $\cM_1\oplus_K \cM_2$
if and only if it is satisfied for $\cM'_1\oplus_K \cM_2$.

In the next lemma, we show that the number of $(\psi,\sigma)$-equivalence classes of
$\sigma$-interpreted $K$-halves is finite
for every MSOL formula $\psi$ and every $K$-signature $\sigma$ of $\psi$.

\begin{lemma}
\label{lm-equiv}
Let $\psi$ be a fixed MSOL formula and $\sigma$ a $K$-signature of $\psi$.
The number of $(\psi,\sigma)$-equivalence classes of $K$-halves is finite.
\end{lemma}

\begin{proof}
We prove the statement of the lemma by induction on the size of $\psi$.
Let us start with the simplest possible formulas,
i.e., $\xi_1=\xi_2$, $\xi_1\in\xi_2$ and $\ind_\cM(\xi_1)$
where $\ind_{\cM}$ is the independence predicate for a matroid $\cM=\cM_1\oplus_K \cM_2$.

Let $\psi$ be equal to $\xi_1=\xi_2$. If
$\sigma(1)\not=\sigma(2)$, then $\psi$ cannot be
$\sigma$-satisfied and thus all $\sigma$-interpreted $K$-halves
are $(\psi,\sigma)$-equivalent. If $\sigma(1)=\sigma(2)\not=\star$
and $\xi_1$ and $\xi_2$ are element variables, then two
$\sigma$-interpreted $K$-halves $\cM_1$ and $\cM'_1$ are
$(\psi,\sigma)$-equivalent if their interpretations assign $\xi_1$
and $\xi_2$ the same element. Hence, the number of
$(\psi,\sigma)$-equivalence classes of $\sigma$-interpreted
$K$-halves is two. If $\sigma(1)=\sigma(2)=\star$, then all
$\sigma$-interpreted $K$-halves are $(\psi,\sigma)$-equivalent
since the truth of $\psi$ is determined by the
$\sigma$-interpretation of $\cM_2$. If $\xi_1$ and $\xi_2$ are set
variables, then a little more complex argument shows that the
number of $(\psi,\sigma)$-equivalence classes of
$\sigma$-interpreted $K$-halves is two if $\sigma(1)=\sigma(2)$
and it is one if $\sigma(1)\not=\sigma(2)$.

For the formula $\psi=\left(\xi_1\in\xi_2\right)$, the number of equivalence classes for $\sigma(1)=\star$ is one as
the truth of the formula $\psi$ depends only on the $\sigma$-interpretation of $\cM_2$ and
the number of $(\psi,\sigma)$-equivalence classes for $\sigma(1)\not=\star$ is two depending
on the fact whether $\xi_1\in\xi_2$ in the $\sigma$-interpretation of $\cM_1$.
For the formula $\psi=\ind_{\cM}(\xi_1)$, there are two $(\psi,\sigma)$-equivalence classes
distinguishing $\sigma$-interpreted $K$-halves with sets $\xi_1$ independent in the $\sigma$-interpretation and
those with a set $\xi_1$ that is not independent.

We now consider more complex MSOL formulas. By standard logic manipulation, we can assume that
$\psi$ is of one of the following forms: $\neg \psi_1$, $\psi_1\lor\psi_2$ or $\exists\xi \psi_1$.
The $(\psi,\sigma)$-equivalence classes for $\psi=\neg\psi_1$ are the same as
the $(\psi_1,\sigma)$-equivalence classes. Hence, their number is finite by the induction.

For the disjunction $\psi_1\lor\psi_2$, if $\cM_1$ and $\cM'_1$ are $\sigma$-interpreted $K$-halves that
are both $(\psi_1,\sigma)$-equivalent and $(\psi_2,\sigma)$-equivalent,
then $\cM_1$ and $\cM'_1$ are also $(\psi,\sigma)$-equivalent. Indeed, if $\cM_1\oplus_K \cM_2$ satisfies
$\psi_j$, $j\in\{1,2\}$, then $\cM'_1\oplus_K \cM_2$ also satisfies $\psi_j$. On the other hand, if $\cM_1\oplus_K \cM_2$
satisfies neither $\psi_1$ nor $\psi_2$, then $\cM'_1\oplus_K \cM_2$ also satisfies neither $\psi_1$ nor $\psi_2$.
Hence, the $(\psi,\sigma)$-equivalence classes of $\sigma$-interpreted $K$-halves are unions
of the intersections of the $(\psi_1,\sigma)$-equivalence classes and $(\psi_2,\sigma)$-equivalence classes.
In particular, the number of $(\psi,\sigma)$-equivalence classes of $\sigma$-interpreted $K$-halves is finite.

The final type of formulas we consider are those of the type $\psi=\exists\xi\psi_1$.
Let $\cM_1$ and $\cM'_1$ be two $\sigma$-interpreted $K$-halves. Note that the domain of $\sigma'$
contains in addition to $\sigma$ the index corresponding to the variable $\xi$.
There are $K+1$ or $K+2$ possible extensions of $\sigma$ to $\sigma'$ (the number depends
on the fact whether $\xi$ is an element or set variable). Let $\sigma'_1,\ldots,\sigma'_N$, $N\in\{K+1,K+2\}$,
be the extensions of $\sigma$ and let $C_i$, $i=1,\ldots,N$, be the number
of $(\psi,\sigma'_i)$-equivalence classes of $\sigma'_i$-interpreted $K$-halves.
Every $\sigma$-interpretation of $\cM_1$ can be extended to a $\sigma'_i$-interpretation of $\cM_1$
in possibly $C_i$ non-equivalent ways.

We claim that if the extensions of the $\sigma$-interpretation of
$\cM_1$ and $\cM'_1$ to $\sigma'_i$-interpretations appear in
exactly the same $(\psi,\sigma'_i)$-equivalence classes for every
$i$, then $\cM_1$ and $\cM'_1$ are $(\psi,\sigma)$-equivalent.
Since there are $2^{C_1+\cdots+C_N}$ possible types of extensions
of the $\sigma$-interpretation to $\sigma'_i$-interpretations and
those $\cM_1$ and $\cM'_1$ with the extensions of the same type must be
$(\psi,\sigma)$-equivalent, the number of
$(\psi,\sigma)$-equivalence classes is bounded by
$2^{C_1+\cdots+C_N}$.

Assume that $\cM_1$ and $\cM'_1$ with
extensions of the same type are not $(\psi,\sigma)$-equivalent. By
symmetry, we can assume that there exist a $\sigma$-interpreted
$\cM_2$ such that $\psi$ is satisfied in $\cM_1\oplus_K \cM_2$ and
is not satisfied in $\cM'_1\oplus_K \cM_2$. Let $\xi_0$ be the
choice of $\xi$ that satisfies $\psi_1$ in $\cM_1\oplus_K \cM_2$.
The value $\xi_0$ uniquely determines an extension of $\sigma$ to
$\sigma'_i$ by setting the value of $\sigma'_i$ corresponding to
$\xi$ to $\varphi_{\cM_1}(\xi_0)$ or $\varphi_{\cM_1}(\{\xi_0\})$
unless $\xi_0$ is an element of $\cM_2$; if $\xi_0$ is an element
of $\cM_2$, the value of $\sigma'_i$ corresponding to $\xi$ should
be set to $\star$. Since $\psi_1$ is satisfied for the
$\sigma$-interpretations of $\cM_1$ and $\cM_2$ by choosing
$\xi=\xi_0$, $\psi_1$ is satisfied for a
$\sigma'_i$-interpretation of $\cM_1$ and a
$\sigma'_i$-interpretation of $\cM_2$. By our assumption, there
exists a $\sigma'_i$-interpretation of $\cM'_1$ that extends its
$\sigma$-interpretation and that is $(\psi,\sigma'_i)$-equivalent
to the $\sigma'_i$-interpretation of $\cM_1$. By the definition of
$(\psi,\sigma'_i)$-equivalence, the formula $\psi$ is satisfied in
$\cM'_1\oplus_K \cM_2$. In particular, there exists a choice of
the value of $\xi$ in $\cM'_1\oplus_K \cM_2$ such that
$\cM'_1\oplus_K \cM_2$ is satisfied for the
$\sigma$-interpretations of $\cM'_1$ and $\cM_2$ which contradicts
our assumption that $\psi$ is not satisfied in $\cM'_1\oplus_K
\cM_2$.
\end{proof}

Let us return to our original problem of deciding MSOL formulas
with no free variables. Analogously, for an MSOL formula $\psi$
with no free variables, two $K$-halves $\cM_1$ and $\cM_1'$ are
{\em $\psi$-equivalent} if the formula $\psi$ is satisfied for
$\cM_1\oplus_K \cM_2$ if and only if $\psi$ is satisfied for
$\cM'_1\oplus_K \cM_2$. By Lemma~\ref{lm-equiv}, the number of
$\psi$-equivalence classes of $K$-halves is finite.

For a $K$-decomposition $\cT$ of a matroid $\cM$, we can obtain a
$K$-half by restricting $\cM$ to the elements corresponding to the
leaves of a subtree of $\cT$ (note that the subsets of the
elements not corresponding to the leaves of $\cT$ can be assigned
non-negative integers as in the definition of a $K$-halved matroid
by Lemma~\ref{lm-aux}). The $K$-half obtained from $\cM$ by
restricting it to the elements corresponding to the leaves of a
subtree rooted at a vertex $v$ of $\cT$ is further denoted by
$\cM_v$. The next lemma is the core of the linear-time algorithm
for deciding the satisfiability of the formula $\psi$.

\begin{lemma}
\label{lm-parse}
Let $\psi$ be a monadic second order logic formula,
let $\cT$ be a $K$-decomposition of a matroid $\cM$ and
let $v$ be a node of $\cT$ with children $v_1$ and $v_2$.
The $\psi$-equivalence class of $\cM_v$ is uniquely determined by the $\psi$-equivalence classes of $\cM_{v_1}$ and
$\cM_{v_2}$, and the functions $\varphi_v$ and $\varphi^r_v$.
\end{lemma}

\begin{proof}
If the statement of the lemma is false, then there exist $\cM_{v_1}$ and $\cM'_{v_1}$ of the same $\psi$-equivalence class
such that the equivalence classes of $\cM_v$ and $\cM'_v$ are different though $\cM_{v_2}$, $\varphi_v$ and $\varphi^r_v$
are the same.
Hence, there exists $\cM_0$ such that
$\psi$ is satisfied for one of the matroids $\cM_v\oplus_K \cM_0$ and $\cM'_v\oplus_K \cM_0$ but not both.
Let $\cM$ be the matroid such that $\cM_{v_1}\oplus_K \cM=\cM_v\oplus \cM_0$ and $\cM'_{v_1}\oplus_K \cM=\cM_v\oplus_K \cM_0$;
such $\cM$ exists since it is uniquely determined by $\cM_0$, $\cM_{v_2}$, $\varphi_v$ and $\varphi^r_v$.
By the choice of $\cM_0$, $\psi$ is satisfied for one of the matroids $\cM_{v_1}\oplus_K \cM$ and $\cM'_{v_1}\oplus_K \cM$
but not both and thus $\cM_{v_1}$ and $\cM'_{v_1}$ cannot be $\psi$-equivalent as supposed.
\end{proof}

We are now ready to present the main result of this section.

\begin{theorem}
\label{thm-msol}
Let $\psi$ be a fixed monadic second order logic and $K$ a fixed integer.
There exists an $O(n)$-time algorithm that given an $n$-element matroid $\cM$ with its $K$-decomposition
decides whether $\cM$ satisfies $\psi$.
\end{theorem}

\begin{proof}
The core of our algorithm is Lemma~\ref{lm-parse}. Since $\psi$ and $K$ is fixed and
the number of $\psi$-equivalence classes of $K$-halves is finite by Lemma~\ref{lm-equiv},
we can wire in the algorithm the transition table from the equivalence classes
of $\cM_{v_1}$ and $\cM_{v_2}$, $\varphi_v$ and $\varphi^r_v$ to the equivalence
class of $\cM_v$ where $v$ is a node of $\cT$ and $v_1$ and $v_2$ its two children.
At the beginning, we determine the equivalence classes of $\cM_v$
for the leaves $v$ of $\cT$; this is easy since the equivalence class of $\cM_v$
for a leaf $v$ depends only on the fact whether the element corresponding to $v$
is a loop or not.

Then, using the wired in transition table, we determine in constant time
the equivalence class of $\cM_v$
for each node $v$ based on $\varphi_v$, $\varphi^r_v$ and the equivalence classes
of $\cM_{v_1}$ and $\cM_{v_2}$ for children $v_1$ and $v_2$ of $v$. Once, the equivalence
classes of $\cM_{v_1}$ and $\cM_{v_2}$ for the two children $v_1$ and $v_2$ of the root of $\cT$
are found, it is easy to determine whether $\cM$ satisfies $\psi$.

As $K$ and $\psi$ are fixed,
the running time of our algorithm is clearly linear in the number of nodes of $\cT$
which is linear in the number of elements of the matroid $\cM$.
\end{proof}

Corollary~\ref{cor-construct} and Theorem~\ref{thm-msol}
yield the following result originally proved by Hlin{\v e}n{\'y}~\cite{bib-hlineny-msol}:

\begin{corollary}
\label{cor-msol}
Let $\FF$ be a fixed finite field, $k$ a fixed integer and $\psi$ a fixed monadic second order logic formula.
There is a polynomial-time algorithm for deciding whether a matroid of branch-width at most $k$
given by its representation over the field $\FF$ satisfies $\psi$.
\end{corollary}

\section{Deciding representability}
\label{sect-repr}

We adopt the algorithm presented in~\cite{bib-stacs}. Let us start with recalling
some definitions from~\cite{bib-stacs}. A \first{rooted branch-decomposition} of $\cM$
is obtained from a branch-decomposition of $\cM$ by subdividing an edge and 
rooting the decomposition tree at a new vertex. If $\cM$ is a matroid and
$(A,B)$ is a partition of its ground set, then two subsets $A_1$ and $A_2$
are \first{$B$-indistinguishable}
if the following identity holds for every subset $B'$ of $B$:
$$r(A_1\cup B')-r(A_1)=r(A_2\cup B')-r(A_2)\;\mbox{.}$$
Clearly, the relation of being $B$-indistinguishable is an equivalence relation
on subsets of $A$.
Finally, the \first{transition matrix} for an inner node $v$ of
a rooted branch-decomposition $\cM$ is the matrix whose rows correspond to
$\overline{E_1}$-indistinguishable subsets of $E_1$ and
columns to $\overline{E_2}$-indistinguishable subsets of $E_2$
where $E_1$ and $E_2$ are the elements corresponding to the leaves of the two subtrees
rooted at the children of $v$ and $\overline{E_1}$ and $\overline{E_2}$ are their complements.
The entry of $\cM$ in the row corresponding to the equivalence class of $F_1\subseteq E_1$ and
in the column corresponding to the equivalence class of $F_2\subseteq E_2$ is
equal to $r(F_1)+r(F_2)-r(F_1\cup F_2)$. By the definition of indistinguishability,
the value of the entry is independent of the choice of $F_1$ and $F_2$ in their
equivalence classes.

The main algorithmic result of \cite{bib-stacs} can be reformulated as follows:

\begin{theorem}
\label{thm-repr}
Let $k$ be a fixed integer and $q$ a fixed prime power. There is a polynomial-time
algorithm that for a matroid $\cM$ given by its rooted branch-decomposition
with transition matrices whose number of rows and columns is at most $k$ and
a (oracle-given) mapping of subsets to equivalence classes corresponding
to rows and columns of its matrices decides whether $\cM$ can be represented
over $\GF(q)$ and if so, it computes one of its representations over $\GF(q)$.
Moreover, the algorithm can be modified to count all non-isomorphic representations
of $\cM$ over $\GF(q)$ and to list them (in time linearly dependent on the number of
non-isomorphic representations).
\end{theorem}

Let $\cT$ be a $K$-decomposition of a matroid $\cM$, $v$ an inner node of $\cT$ and
$E_v$ the subset of the ground set of $\cM$ containing the elements corresponding
to the leaves of subtree of $\cT$ rooted at $v$.
View $\cT$ as a rooted branch-decomposition of $\cM$ (we do not claim any upper
bound on its width here).
Observe that any two subsets of $E_v$
assigned the same color at $v$ are $\overline{E_v}$-indistinguishable
where $\overline{E_v}$ is the complement of $E_v$.
In addition, the values of the function $\varphi^r_v$ are entries of the transition
matrix at $v$ as defined in the beginning of this section. Hence, Theorem~\ref{thm-repr}
yields the following.

\begin{corollary}
\label{cor-repr}
For every integer $K$ and every prime power $q$,
there is a polynomial-time algorithm that for a matroid $\cM$ given by its $K$-decomposi\-tion,
decides whether $\cM$ can be represented over $\GF(q)$ and
if so, it computes one of its representations over $\GF(q)$.
Moreover, the algorithm can be modified to count all non-isomorphic representations
of $\cM$ over $\GF(q)$ and to list them (in time linearly dependent on the number of
non-isomorphic representations).
\end{corollary}

\section*{Acknowledgement}

The author would like to thank Bruno Courcelle for his valuable comments
on the conference paper~\cite{bib-stacs} and 
Ond{\v r}ej Pangr{\'a}c for his insightful comments on the content of this paper as well as
careful reading its early version.

\end{document}